\documentclass[onecolumn,english,12pt]{IEEEtran}

\usepackage[square,numbers]{natbib}
\usepackage{subcaption}
\usepackage{microtype}
\usepackage{graphicx}
\usepackage{graphics}
\usepackage{wrapfig}
\usepackage{float}
\usepackage{caption}
\usepackage{booktabs} 
\usepackage{rotating}
\usepackage{threeparttable}

\usepackage{hyperref}

\usepackage{amsmath}
\usepackage{amssymb}
\usepackage{mathtools}
\usepackage{amsthm}
\usepackage{algorithmicx}
\usepackage{algorithm}
\usepackage[noend]{algpseudocode}
\usepackage[dvipsnames]{xcolor}
\usepackage[capitalize,noabbrev]{cleveref}

\theoremstyle{plain}
\newtheorem{theorem}{Theorem}
\newtheorem{proposition}[theorem]{Proposition}
\newtheorem{lemma}[theorem]{Lemma}

\theoremstyle{definition}

\theoremstyle{remark}

\usepackage[textsize=tiny]{todonotes}

\title{Out-of-Distribution Detection using Maximum Entropy Coding}

\author{Mojtaba Abolfazli, Mohammad Zaeri Amirani, Anders H{\o}st-Madsen, June Zhang, Andras Bratincsak\thanks{M. Abolfazli, M. 
Zaeri Amirani, A. H{\o}st-Madsen, and J. Zhang are with the Department of Electrical and Computer Engineering, University of Hawaii at Manoa, 2540 Dole Street, Honolulu, HI 96822; e-mail: \{mojtaba,ahm,zjz\}@hawaii.edu. 

A. Bratincsak is with the Department of Pediatrics, John A. Burns School of Medicine, University of Hawaii, Honolulu, HI 96813; e-mail: andrasb@hphmg.org. 

The research was funded in part by the NSF grant CCF-1908957  and NIH grant 1202518S0.}}

\begin{document}
\maketitle



\begin{abstract}
Given a default distribution $P$ and a set of test data
$x^M=\{x_1,x_2,\ldots,x_M\}$ this paper seeks to answer
the question if it was likely that $x^M$ was generated by $P$.
For discrete distributions, the definitive answer is in principle given
by Kolmogorov-Martin-L\"{o}f randomness. In this paper we seek to
generalize this to continuous distributions. We consider a set
of statistics $T_1(x^M),T_2(x^M),\ldots$. To each statistic we
associate its maximum entropy distribution and with this
a universal source coder. The maximum entropy distributions are
subsequently combined to give a total codelength, which is compared
with $-\log P(x^M)$. We show that this approach satisfied a number
of theoretical properties.

For real world data $P$ usually is unknown. We transform data into
a standard distribution in the latent space using a bidirectional
generate network and use maximum entropy coding there. 
We compare the resulting method
 to other methods that also used generative neural networks to detect anomalies. In most cases, our results show better performance.
\end{abstract}


\section{Introduction}
We consider the following problem. Given a \emph{default} distribution
$P$ (which could be continuous or discrete), and
a set of  test data $x^M=\{x_1,x_2,\ldots,x_M\}$ (which for now
does not have to be IID), we
would like to determine if it is likely that $x^M$ was
generated by $P$ or by another distribution. For (binary) discrete data,
this problem was solved theoretically by Kolmogorov and Martin-L\"{o}f
through Kolmogorov complexity \cite{LiVitanyi}. The starting point is what
is called a P-test, which can be thought of as testing
a specific data statistic (e.g., is the mean the correct
one according to $P$).
There are many such statistics, and a universal test is one
that includes all statistics. Martin-L\"{o}f showed that
a universal (sum) P-test is given by $K(x^M|M)<M$ when $P$ is uniform, 
where $K$ is Kolmogorov complexity. Replacing Kolmogorov
complexity (which is uncomputable) with a universal source coder,
this was used to develop Atypicality in \cite{HostSabetiWalton15}.

In this paper we consider the following specific problem.
We start with a
continuous distribution $P$ over $\mathbb{R}^n$ and IID test
data $\mathbf{x}^M = \{\mathbf{x}_i \in \mathbb{R}^n\}, i = 1, \ldots, M >1$. The question is if $\mathbf{x}^M$ is likely to
have been generated by $P$.
This problem is known as out-of-distribution (OOD) detection and is also called group anomaly
detection (GAD). 
For this setup, Kolmogorov complexity cannot be directly applied.
Our aim in this paper is
to still use the principles of Martin-L\"{o}f
randomness \cite{LiVitanyi} to develop principled methods. 
We base this on statistics of the data, to which we associate
maximum entropy distributions, which can in turn be used for coding.


\subsection{Related Work}
\label{sec:RelatedWork}
There have been a number of works on OOD or GAD, and we will
just discuss a few.
In statistics there are for example the Kolmogorov-Smirnoff (KS) test \cite{corder2014nonparametric} and the Pearson $\chi^2$ test \cite{lehmann1986testing}.

In machine learning
  one-class learning has been used \cite{chandola2009anomaly, muandet2013one}.
Another approach is using probabilistic generative models and considering the OOD problem in the latent domain \cite{rabanser2019failing}.
Reference \cite{NalisnickAl19} proposed a method based on the  empirical entropy.
\cite{chalapathy2018group, zhang2020towards} used AAE and VAE neural networks to transform the data.

Many of the ML methods simply directly or indirectly
use likelihood for OOD, i.e., how likely was it
that data came from $P$? However, likelihood is not a good
measure. As an example, suppose that $P$ is the uniform IID distribution
over $[0,1]$. Then any sequence $x^M$ is equally likely,
i.e., nothing is OOD according to likelihood. 
The statistical tests for OOD therefore use the samples
$x^M$ to generate an alternative distribution. 
In the KS test \cite{corder2014nonparametric},
the empirical CDF of $x^M$ is compared with the
CDF of $P$. In the Pearson $\chi^2$ test \cite{lehmann1986testing}, the empirical
PMF of $x^M$ is compared with the PMF of $P$. This shows that
OOD detection has to be generative, in the sense of coming up with
an alternative distribution for the OOD data, and this is also consistent
with Kolmogorov Martin-L\"{o}f
randomness. 

It is difficult to extend KS test to higher dimensions since computing the empirical CDFs depends on the arrangements of dimensions. Some methods were proposed in \cite{fasano1987multidimensional, hagen2021accelerated}. 
As opposed to that, our method can be used for high-dimensional
data.

The method in this paper is related to Rissanen's minimum description length (MDL) \cite{Rissanen78,GrunwaldBook}.
We have used this to find transients in
\cite{SabetiHost17} and it was used for change point
detection in \cite{yamanishi2018model, yamanishi2021change}.
The problem and methodology in this paper are different, so
these papers are not directly applicable.

\section{Methodology}\label{sec:methodology}
There is no true generalization of universal source coding used for
Atypicality in \cite{HostSabetiWalton15}
in
the real case, but we will still maintain the idea of coding.
For assuming that $\mathbf{x}^M$ comes from the default distribution, $P$, the codelength of the test set as argued by Rissanen \cite{Rissanen86} is
\begin{align*}
  L_P(\mathbf{x}^M)&=-\log P(\mathbf{x}^M)=-\log \prod_{i=1}^M P(\mathbf{x}_i).
\end{align*}
We would like to
compare this with a ``universal" codelength.

Our starting point is Martin-L\"{o}f's idea of a P-test \cite{LiVitanyi}. We consider
a statistic $\mathbf{T}:\mathbb{R}^n\to\mathbb{R}^k$ with
$\hat {\mathbf{t}}=\frac 1 M\sum_{i=1}^M {\mathbf{T}}(\mathbf{x}_i)$. If
$\hat {\mathbf{t}}\neq E_P[{\mathbf{T}}(\mathbf{x})]$ one could consider the test data
OOD. In order to put this both in a likelihood ratio test framework and a coding
framework, we need to associate an alternative distribution with the statistic ${\mathbf{T}}$. The natural
choice for such a distribution is the maximum entropy distribution 
which has the form \cite{CoverBook}
\begin{align}
  P_T(\mathbf{x})&=\exp\left(\boldsymbol{\lambda}^T {\mathbf{T}}(\mathbf{x})-A(\boldsymbol{\lambda})\right) \label{eq:maxent}
\end{align}
Here $A(\boldsymbol{\lambda})$ is
the log-partition functions ensuring
$P_T$ integrates to one. The statistic therefore has to be
one that has a valid maximum entropy distribution. While the maximum
entropy distribution seems reasonable, it also has the following straighforward
optimality property


\begin{proposition}\label{thm:maxentmaxmin}
	Consider the set of distributions $P'$ that satisfy 
	$E_{P'}[{\mathbf{T}}(\mathbf{x})]={\mathbf{t}}$. Among those the maximum
	entropy distribution $P_T$ is the minimax coding distribution,
	i.e., it achieves
\begin{align*}
	\inf_{P:E_P[{\mathbf{T}}(x)]={\mathbf{t}}}\sup_{P':E_{P'}[{\mathbf{T}}(x)]={\mathbf{t}}} E_{P'}[-\log P(\mathbf{x})]
\end{align*}
\end{proposition}
\begin{proof}
	It is clear from (\ref{eq:maxent}) that $E_{P'}[-\log P(\mathbf{x})]$ is independent of $P'$. If there were
some other distribution $\tilde P$ that had lower minimax
coding length than $P_T$ for all $P'$, it would therefore have to be strictly smaller for
all $P'$. But $P_T$ is optimum if the data is actually generated
by $P_T$, so this is not possible.
\end{proof}
Maximum entropy is actually commonly used in universal source coding. One
method for universal source coding (of discrete sources) is to first transmit
the type of the sequence (i.e., the histogram) and then the specific sequence
with that type (e.g., \cite[Section 13.2]{CoverBook} -- assuming a uniform distribution over sequences, which is
precisely maximum entropy.

The importance of this property can be seen as follows. The log-likelihood ratio
test with $P_T$ can be written as
\begin{align*}
  L_P(\mathbf{x}^M)-L_{P_T}(\mathbf{x}^M)\geq \tau 
\end{align*}
To the first order, the detection performance depends on
$E_{P'}[L_P(\mathbf{x})]-E_{P'}[L_{P_T}(\mathbf{x})]$.
According to Proposition \ref{thm:maxentmaxmin} the 
second term is independent of $P'$. In many cases the
first term is also independent of $P'$. For example
if $P$ is uniform  over $[0,1]$ or if $P$ is
itself a maximum entropy distribution for a subset of
the statistic ${\mathbf{T}}$. In those cases, using maximum entropy
results in a detection performance independent of $P'$
and only dependent on ${\mathbf{T}}$ -- to the first order.

It is natural to think that a more complex statistic 
will be able to capture more types
of deviation from the default distribution. But it might not be
better for OOD detection, as the following theorem shows. Consider a
maximum entropy distribution $P_T$ and suppose that the default
distribution $P=P_{t_0}=P_{T}(\mathbf{x};{\mathbf{T}}=\mathbf{t}_0)$. Set a
desired false alarm
probability $\alpha=P_{FA}$ and detection probability
$\beta=P_D$. Let $\mathcal{S}_{\alpha,\beta}(M)$ be the set of
distributions $P_t=P_{T}(\mathbf{x};{\mathbf{T}}=\mathbf{t})$ that can
be detected with $P_{FA}\leq\alpha$ and $P_D\geq\beta$ with $M$ samples. The question is
how little deviation from the default distribution is needed for detection;
we measure this by the radius $\inf_{P_t\in S_{\alpha,\beta}(M)} D(P_t\|P_{t_0})$, where
$D(\cdot\|\cdot)$ is relative entropy. This radius can be calculated
asymptotically as $M\to\infty$,
To prove it, we need the following Lemma

\begin{lemma}\label{thm:chi2}
Let ${\mathbf{T}}$ be a statistic and $P_T$ the corresponding maximum
entropy distribution. Suppose that data is generated according to
$P_T(\mathbf{x};{\mathbf{T}}={\mathbf{t}})$ and let $\hat {\mathbf{t}}=\frac 1 M\sum_{i=1}^MT(\mathbf{x}_i)$.
Then
\begin{align*}
  2\ln 2(-\log P(\mathbf{x}^M;{\mathbf{T}}={\mathbf{t}})
  +\log P(\mathbf{x}^M;{\mathbf{T}}=\hat {\mathbf{t}})\stackrel{D}{\to }\chi^2_m
\end{align*}
\end{lemma}
\begin{proof}
Maximum entropy distributions (\ref{eq:maxent}) are in the exponential family. Notice that $\hat {\mathbf{t}}=\frac 1 M\sum_{i=1}^M {\mathbf{T}}(\mathbf{x}_i)$ is the maximum likelihood estimator for this model through
some transformation $\hat {\mathbf{t}}\mapsto\hat{\boldsymbol{\lambda}}$ \cite{MoulinVeeravalliBook}. We also have \cite{MoulinVeeravalliBook}
\begin{align*}
\sqrt{M}\left(\hat{\boldsymbol{\lambda}}-{\boldsymbol{\lambda}}\right)
  \stackrel{D}{\to}N(0,\mathbf{J}^{-1})
\end{align*}
where $\mathbf J = \mathbf J(\boldsymbol{\lambda})$ is the Fisher information matrix.
Let us use the notation $\mathbf{J}(\hat{\boldsymbol{\lambda}}) = \nabla_{\boldsymbol{\lambda}}^2 \left(A(\boldsymbol{\lambda})\right)|_{\boldsymbol{\lambda}=\hat{\boldsymbol{\lambda}}}$, so:
\begin{align*}
    E\left[\mathbf{J}(\hat{\boldsymbol{\lambda}})\right] &= -\nabla_{\boldsymbol{\lambda}}^2\ln{P_T(x;\boldsymbol{\lambda})} = \nabla_{\boldsymbol{\lambda}}^2 A(\boldsymbol{\lambda}) = \mathbf{J}
\end{align*}
Now
\begin{align*}
  R&=-\ln P(\mathbf{x}^M;{\mathbf{T}}={\mathbf{t}})
  +\ln P(\mathbf{x}^M;{\mathbf{T}}=\hat {\mathbf{t}}) \nonumber\\
  &= \boldsymbol{\lambda}^TT(\mathbf{x}^M)-MA(\boldsymbol{\lambda})
  -\hat{\boldsymbol{\lambda}}^TT(\mathbf{x}^M)+MA(\hat{\boldsymbol{\lambda}})
\end{align*}
We  expand $R$ as a Taylor series around $\hat{\boldsymbol{\lambda}}$.
Notice that $\left.\frac{\partial R}{\partial \boldsymbol{\lambda}}\right|_{\boldsymbol{\lambda}=\hat{\boldsymbol{\lambda}}}=0$ because
$\hat{\boldsymbol{\lambda}}$ is the MLE. Thus, 
\begin{align*}
  R &= -\frac M 2 \left(\hat{\boldsymbol{\lambda}}-{\boldsymbol{\lambda}}\right)^T\mathbf{J}(\hat{\boldsymbol{\lambda}})
  \left(\hat{\boldsymbol{\lambda}}-{\boldsymbol{\lambda}}\right)
  +o(\boldsymbol{\lambda}^2) 
\end{align*}

According to \cite[p. 271]{MoulinVeeravalliBook}
$\mathbf J(\boldsymbol{\lambda})$ is a continuous
function of $\boldsymbol{\lambda}$.
As $\hat{\boldsymbol{\lambda}} \stackrel{p}{\to}\boldsymbol{\lambda}$ then
 $\mathbf J(\hat{\boldsymbol{\lambda}}) \stackrel{p}{\to}\mathbf{J}$ by Slutsky's theorem \cite{SerflingBook}.
Therefore, also by Slutsky's theorem, the first term in $R$ is $\chi^2$. Finally, applying the univariate delta method on the function:
\begin{equation*}
    g(\hat{\boldsymbol{\lambda}}) = R + \frac{M}{2} \left(\hat{\boldsymbol{\lambda}}-{\boldsymbol{\lambda}}\right)^T\mathbf{J}(\hat{\boldsymbol{\lambda}})
  \left(\hat{\boldsymbol{\lambda}}-{\boldsymbol{\lambda}}\right)
\end{equation*}
results in $g(\hat{\boldsymbol{\lambda}}) \stackrel{p}{\to} 0$, which proves the Lemma.

%

\end{proof}

\begin{theorem}\label{thm:complexity}
Let $P_T$ be a maximum entropy distribution.
Consider detection between $\mathbf{T}=\mathbf{t}_0$ and $\mathbf{T}\neq{\mathbf{t}_0}$. Fix the false alarm
probability $\alpha=P_{FA}$ and the detection probability
$\beta=P_D$ as $M\to\infty$. 
Let $\mathcal{S}_{\alpha,\beta}(M)$ be the set of
distributions $P_t=P_{T}(\mathbf{x};{\mathbf{T}}=\mathbf{t})$ that can
be detected with $P_{FA}\leq\alpha$ and $P_D\geq\beta$ with $M$ samples..
Then
\begin{align*}
  &\lim_{M\to\infty} \inf_{P_t\in S_{\alpha,\beta}(M)} M\frac{1}{2\ln 2}D(P_t\|P_{t_0}) = F_{\chi^2_k}^{-1}(\beta)-F_{\chi^2_k}^{-1}(\alpha)
\end{align*}
where $F_{\chi^2_k}$ is the CDF of the $\chi^2$ distribution with
$k$ degrees of freedom.
\end{theorem}
\begin{proof}
Let
\begin{align*}
  R&=2\ln 2(-\log P(\mathbf{x}^M;{\mathbf{T}}={\mathbf{t}})
  +\log P(\mathbf{x}^M;{\mathbf{T}}=\hat {\mathbf{t}}))
\end{align*}
Then according to Lemma \ref{thm:chi2} the false alarm probability satisfies
\begin{align}
  \lim_{M\to\infty}P_{FA}&=\lim_{M\to\infty}P(R>\tau)=1-F_{\chi^2_k}(\tau) \label{eq:PFAp}
\end{align}

Consider a sequence of distributions $P_{t_M}(\mathbf{x})=P_T(\mathbf{x};{\mathbf{T}}=t_M)$ that
satisfy
\begin{align}
  \lim_{M\to\infty}MD(P_{t_M}\|P_{t_0}) &= K \label{eq:MDlimit}
\end{align}
for some constant $K$. For a maximum entropy distribution, we have
\begin{align*}
  MD(P_{t_M}\|P_{t_0})&=M(\boldsymbol{\lambda}-\boldsymbol{\lambda}_0)^T
  E_{\boldsymbol{\lambda}}[{\mathbf{T}}(\mathbf{x}]+M(A(\boldsymbol{\lambda}_0)-A(\boldsymbol{\lambda}))
\end{align*}
In order for (\ref{eq:MDlimit}) to be satisfied, we must then
have 
\begin{align}
	\lim_{M\to\infty}M(\boldsymbol{\lambda}-\boldsymbol{\lambda}_0)=\mathbf{v} \label{eq:MDlimit2}
\end{align}
for some vector $\mathbf{v}$.
 Now write
\begin{align*}
  R &= 2\ln 2\left(-\log P_{t_0}(\mathbf{x}^M)+\log P_{t_M}(\mathbf{x}^M)\right.\nonumber\\
  &\left.-\log P_{t_M}(\mathbf{x}^M)+\log P_{\hat {\mathbf{t}}}(\mathbf{x}^M)\right) 
\end{align*}
The first term can be written as
\begin{align*}
 \MoveEqLeft -\ln P_{t_0}(\mathbf{x}^M)+\ln P_{t_M}(\mathbf{x}^M)\nonumber\\
  &= (\boldsymbol{\lambda}-\boldsymbol{\lambda}_0)^T
  \sum_{i=1}^MT(\mathbf{x}_i)+M(A(\boldsymbol{\lambda}_0)-A(\boldsymbol{\lambda})) \nonumber\\
  &=M(\boldsymbol{\lambda}-\boldsymbol{\lambda}_0)^T
  \frac 1 M\sum_{i=1}^MT(\mathbf{x}_i)+M(A(\boldsymbol{\lambda}_0)-A(\boldsymbol{\lambda})) \nonumber\\
  &\stackrel{P}{\to} K
\end{align*}
because of the law of large numbers and (\ref{eq:MDlimit}) and
(\ref{eq:MDlimit2}).
At the same time
\begin{align*}
  2\ln 2\left(-\log P_{t_M}(\mathbf{x}^M)+\log P_{\hat t}(\mathbf{x}^M)\right)\stackrel{D}{\to}\chi_k^2
\end{align*}
according to Lemma \ref{thm:chi2}. There is the difference that the
parameter is not fixed, but examining the proof of Lemma \ref{thm:chi2}
shows that this makes no difference. Then from
Slutsky's theorem
\begin{align*}
  R&\stackrel{D}{\to}\chi^2_k+K
\end{align*}
Thus
\begin{align}
    \lim_{M\to\infty}P_D&=F_{\chi^2_k}(\tau-K) \label{eq:PDp}
\end{align}
The theorem now follows by solving (\ref{eq:PFAp}) and (\ref{eq:PDp}) with respect to $K$ for given $\alpha$ and $\beta$.

\end{proof}
Perhaps a more intuitive measure of distance between distributions
is total variation distance, for continuous distributions
\begin{align*}
  d_\text{TV}(P_{y_M},P_{y_0})=\frac 1 2
  \int |P_{y_M}(x)-P_{y_0}(x)|dx
\end{align*}
Then Pinsker's inequality gives
\begin{align*}
  \lim_{M\to\infty} \inf_{P_{y_M}} \sqrt Md_\text{TV}(P_{t_M},P_{t_0}) &\leq \sqrt{F_{\chi^2_m}^{-1}(\beta)-F_{\chi^2_m}^{-1}(\alpha)}
\end{align*}

What this means is that the closest distribution that can
be detected as OOD is
\begin{align*}
  D(P_{t_M}\|P_{t_0}) &\approx \frac{F_{\chi^2_m}^{-1}(\beta)-F_{\chi^2_m}^{-1}(\alpha)}{M}
\end{align*}
This increases nearly linearly with $m$. Thus, higher complexity
models are more difficult to detect, or more precisely, small
deviations in high complexity models are more difficult to detect.

The conclusion is that one should try to detect OOD
with the simplest statistics possible. Yet, the statistic
also has to be complex enough to capture deviations.
The solution to this dilemma is consider simple and complex
statistics simultaneously, but with a penalty for
more complex models in light of Theorem \ref{thm:complexity}.

The statistics  have to be chosen with the default distribution $P$ in mind.
Intuitively, the statistics have to indicate deviations from $P$ well.
To detect small deviations in distribution, one would like statistics
so that $D(P_T\|P)$ can be made small. One way to obtain this is if $P$
is itself a maximum entropy distribution for some value of ${\mathbf{T}}$ -- then
$D(P_T\|P)$ can be made arbitrarily small. Another possibility is to
have a sequence of statistics $T_i$ so that $D(P_{T_i}\|P)$ can
be made arbitrarily small -- but with a complexity
cost according to Theorem \ref{thm:complexity}. As an example, suppose that the default
distribution is $\chi^2$. If ${\mathbf{T}}(x)=(x,x^2)$ (mean and variance) the
maximum entropy distribution is Gaussian, which is not close to $\chi^2$; 
the consequence is that mean and variance has to change by large amounts
for detection.
But if ${\mathbf{T}}(x)=(x,\ln x)$, the maximum entropy distribution is Gamma, of
which the $\chi^2$ is a special case. 

Using only maximum entropy distributions might seem limiting. However,
the alternative distribution does not necessarily have to be modeled well.
Suppose as an example the default distribution is $U[0,1]$, while data is generated
according to $U[0,\theta]$, which is not maximum entropy. If
$\theta>1$, as soon as some $x>1$ is seen, the default coder will
give a codelength of infinity, and data will be declared OOD. If $\theta<1$,
the histogram distribution described below, Section \ref{sec:histogram} can be used, and this will detect $U[0,\theta]$.

\subsection{Coding}\label{sec.MDL}
Consider a sequence (finite or countable) of statistics $\mathbf{T}_i$ of varying complexity. We would
like to combine all the statistics into a single test. This is similar to
what is done for P-tests in Martin-L\"{o}f randomness \cite{LiVitanyi}. We use a coding
approach, inspired by Kolmogorov complexity and universal source coding
in Atypicality \cite{HostSabetiWalton15}. The encoder and decoder both know the sequence of
possible statistics $\mathbf{T}_i$. The idea is to use these statistics to encode
the sequence $\mathbf{x}^M$ with the shortest codelength possible. 
Consider first a single statistic ${\mathbf{T}}$. One approach is that the
encoder first calculates $\hat {\mathbf{t}}=\frac 1 M\sum_{i=1}^M {\mathbf{T}}(\mathbf{x}_i)$, conveys that to the decoder, and then
encodes $\mathbf{x}^M$ with $P_T$. Since $\hat {\mathbf{t}}$ is real-valued, it
has to be quantized to minimize total codelength. It will be noticed
that is exactly as in Rissanen's minimum description length (MDL) \cite{Rissanen78,GrunwaldBook},
and we can therefore use the rich theory from MDL. For example,
one can use sequential coding instead of the 
two-step coding above. However, our aim is
not to find a good model as in MDL.

Now consider the whole sequence of statistics ${\mathbf T}_i$. One approach
to use the statistic $\mathbf{T}_i$ that results in the shortest codelength.
From a coding point of view, the encoder needs to tell
the decoder which statistic was used. This can be encoded using Elias 
code for the integers \cite{Rissanen83,Elias75}, which uses $\log^* m+c$, where
$\log^* m=\log k+\log\log k+\cdots$, continuing until the argument to
the $\log$ becomes negative, and $c$ is a constant making Kraft's
inequality satisfied with equality.

We can then write the resulting test explicitly as
\begin{align}
  \min_{i}-\log  P'_{T_i}(\mathbf{x}^M)+L(\mathbf{T}_i)+\log^*(i)+\tau
  < -\log  P(\mathbf{x}^M) \label{eq:MDLtest}
\end{align}
where $L(\mathbf{T}_i)$ is the length of the code to encode the (quantized) statistic
$\mathbf{T}_i$ -- we will later describe how precisely the coding is done.
To recap, the coder on the left-hand side works by telling
the decoder which encoder has been used, and then encoding
according to it. However, a more efficient coder can be
obtained by weighting the different coders, the principle
in the Context Tree Weighting (CTW) coder and other coders
\cite{WillemsAl95}. We can then write 
\begin{align}
  -\log\left( \sum_{i=1}^\infty P'_{T_i}(\mathbf{x}^M)2^{-L(\mathbf{T}_i)-\log^*(i)}\right)+\tau 
  < -\log  P(\mathbf{x}^M) \label{eq:WMDL}
\end{align}
We call this \emph{weighting}. Notice that this approach does not
find a model with the shortest description length, and is therefore
distinct from MDL. We will be using (\ref{eq:WMDL}) in our implementation. 

While we know from Theorem \ref{thm:complexity}
that we should consider statistics of varying
complexity, it is not obvious that combining them
using (\ref{eq:MDLtest}) or (\ref{eq:WMDL}) result
in good OOD performance. Theoretical validation
really is only possible asymptotically. The
most meaningful limit is $k\to\infty$ while
$M$ is fixed or $M\ll k$ (if $k$ is fixed
and $M\to\infty$ one can just use the most
complex model without much loss). Not all models
allow $k\ll M$, and we will therefore limit the
analysis to a specific case in the following section.

We will show that in some cases, coding results
in optimum performance in the sense of Theorem
\ref{thm:complexity}. Consider a sequence
of statistics $\mathbf{T}_i$ such that 
$\mathbf{T}_i$ is a subset of the statistics
in $\mathbf{T}_{i+1}$. Suppose that the
default distribution is the maximum entropy
distribution corresponding to $\mathbf{T}_1$,
and suppose the OOD data is generated by
the maximum entropy distribution corresponding to
$\mathbf{T}_k$ for some $k>1$. If $k$ is known,
Theorem \ref{thm:complexity} gives the peformance. As $k$ is not known,
we use (\ref{eq:MDLtest}) or (\ref{eq:WMDL}). 
\begin{theorem}
	Consider the same setup as in Theorem \ref{thm:complexity} and suppose
	(\ref{eq:MDLtest}) or (\ref{eq:WMDL}) is used
	for detection. Assume $L(T_i)=\frac{|\mathbf{T}_i|}{2}\log M$ and ignore $\log^* m$.
	Then the detection performance
	is the same as if $k$ were known and given by 
\begin{align}
  \lim_{M\to\infty} \inf_{P_t\in S_{\alpha,\beta}(M)} \frac M{\log M}\frac{1}{2\ln 2}D(P_t\|P_{t_0}) &\leq \frac {|\mathbf{T}_k|} 2
  \label{eq:kknown}
\end{align}
\end{theorem}
\begin{proof}
Let us assume we consider statistics $\mathbf{T}_i$ for $i\leq K(M)$ where $K(M)$ is an increasing
function of $M$.
We can bound the false alarm probability by the union bound
\begin{align}
	P_{FA}\leq\sum_{i=1}^{K(M)}P\left(
	  -\log P(\mathbf{x^M})+\log P_{T_i}(\mathbf{x}^M)
	  \geq \tau+\frac{|\mathbf{T}_i|}{2}\log M\right)
	 \label{eq:PFAproof}
\end{align}
According to Lemma \ref{thm:chi2} $
	  2\ln 2(-\log P(\mathbf{x^M})+\log P_{T_i}(\mathbf{x}^M))$
is approximately $\chi^2$. We will first evaluate
(\ref{eq:PFAproof}) under some simplifying assumptions:
1)  $
	  2\ln 2(-\log P(\mathbf{x^M})+\log P_{T_i}(\mathbf{x}^M))$ is
exact $\chi^2$, 2) $|\mathbf{T}_i|=i$, 3) $\tau=0$. We
can then use the Chernoff bound on (\ref{eq:PFAproof}) to
get
\begin{align*}
	P_{FA} &\leq \sum_{i=1}^{K(M)}\exp\left(- i \ln M\right)
	\left(\frac{i\ln M}{2i}\right)^{i/2} \nonumber\\
	&= \sum_{i=1}^{K(M)}
	\left(\frac{\sqrt{\ln M}}{\sqrt 2 M}\right)^{i}\nonumber\\
	&\leq \frac{\sqrt{\ln M}}{\sqrt 2 M}\frac 1{1-\frac{\sqrt{\ln M}}{\sqrt 2 M}}
\end{align*}
We notice that $P_{FA}\to 0$ not matter what $K(M)$ is, that is, at some point it will be less than the given $\alpha$. Of course
the $\chi_2$ approximation is not exact. However, we
can always find some $K(M)\to\infty$ so that the $\chi_2$ approximation
is good enough so that still 
$P_{FA}\to 0$. If $|\mathbf{T}_i|>i$ it can only lead to faster
decrease towards zero. Finally, since we only require
$P_{FA}\leq\alpha$, we could decrease $\tau$ below zero. The upside is that we can ensure $P_{FA}\leq\alpha$ for sufficienly
large $M$, even if the assumptions are not satisfied.

Consider a sequence of distributions $P_{t_M}(\mathbf{x})=P_{T_k}(\mathbf{x};{\mathbf{T}_k}=\mathbf{t}_M)$ that
satisfy
\begin{align}
  \lim_{M\to\infty}\frac{M}{\log M}D(P_{t_M}\|P_{t_0}) &= K \label{eq:MDlimitMDL}
\end{align}
for some constant $K$. For a maximum entropy distribution, we have
\begin{align*}
  D(P_{t_M}\|P_{t_0})&=(\boldsymbol{\lambda}-\boldsymbol{\lambda}_0)^T
  E_{\boldsymbol{\lambda}}[{\mathbf{T}}(\mathbf{x}]+(A(\boldsymbol{\lambda}_0)-A(\boldsymbol{\lambda}))
\end{align*}
In order for (\ref{eq:MDlimitMDL}) to be satisfied, we must then
have 
\begin{align}
	\lim_{M\to\infty}\frac M{\log M}(\boldsymbol{\lambda}-\boldsymbol{\lambda}_0)=\mathbf{v} \label{eq:MDlimit2x}
\end{align}
for some vector $\mathbf{v}$.
 Now write
\begin{align*}
  R &= 2\ln 2\left(-\log P_{t_0}(\mathbf{x}^M)+\log P_{t_M}(\mathbf{x}^M)\right.\nonumber\\
  &\left.-\log P_{t_M}(\mathbf{x}^M)+\log P_{\hat {\mathbf{t}}}(\mathbf{x}^M)\right) 
\end{align*}
For the first term we have
\begin{align*}
 \MoveEqLeft \frac 1{\log M}\left(-\ln P_{t_0}(\mathbf{x}^M)+\ln P_{t_M}(\mathbf{x}^M)\right)\nonumber\\
  &= \frac 1{\log M}\left((\boldsymbol{\lambda}-\boldsymbol{\lambda}_0)^T
  \sum_{i=1}^MT(\mathbf{x}_i)+M(A(\boldsymbol{\lambda}_0)-A(\boldsymbol{\lambda}))\right) \nonumber\\
  &=\frac M{\log M}(\boldsymbol{\lambda}-\boldsymbol{\lambda}_0)^T
  \frac 1 M\sum_{i=1}^MT(\mathbf{x}_i)+\frac M{\log M}(A(\boldsymbol{\lambda}_0)-A(\boldsymbol{\lambda})) \nonumber\\
  &\stackrel{P}{\to} K
\end{align*}
because of the law of large numbers and (\ref{eq:MDlimit}) and
(\ref{eq:MDlimit2x}).
At the same time
\begin{align*}
  2\ln 2\left(-\log P_{t_M}(\mathbf{x}^M)+\log P_{\hat t}(\mathbf{x}^M)\right)\stackrel{D}{\to}\chi_k^2
\end{align*}
The detection probability is
\begin{align*}
  P_D \geq P\left(R\geq \frac{|\mathbf{T}_k|}{2}\log M\right)
\end{align*}

Then from
Slutsky's theorem
\begin{align*}
  R&\stackrel{D}{\to}\chi^2_k+K
\end{align*}
Thus
\begin{align}
    \lim_{M\to\infty}P_D&=F_{\chi^2_k}(\tau-K) \label{eq:PDpx}
\end{align}
The theorem now follows by solving (\ref{eq:PFAp}) and (\ref{eq:PDpx}) with respect to $K$ for given $\alpha$ and $\beta$.

\end{proof}

While the above gives some indication that coding works,
the
most meaningful limit is $k\to\infty$ while
$M$ is fixed or $M\ll k$ (if $k$ is fixed
and $M\to\infty$ one can just use the most
complex model without much loss). Not all models
allow $k\ll M$, and we will therefore limit the
analysis to a specific case in the following section.

\subsection{Histogram}\label{sec:histogram}
In this section we will show theoretically the advantage of the coding
approach for combining statistics, limiting ourselves to scalar data for analytical tractability.
We consider the case with the default distribution $P$
uniform over $[0,1]$. Any one-dimensional problem can be transformed
into this by transformation with the CDF \cite{CoverBook}: it is well known that for a continuous random variable
$X$ with CDF $F_X$, 
$U=F_X(X)$ has a uniform distribution on $(0,1)$. We will see later that
such transformations are essential for working with complex distributions.
Thus, this is a  general one-dimensional problem.

We use the following statistic: we divide the interval $[0,1]$ into
$m$ equal length subintervals, and count the number of samples in each
interval. The corresponding maximum entropy distribution is the uniform
distribution over each subinterval. This is of course the histogram of
the data. One notices that the default distribution is the histogram
with $m=1$, so this is a good sequence of statistics for
this problem according to the theory in Section \ref{sec:methodology}.

Let $M\hat p_j$ be the number of samples in the $j$-th interval.
The maximum entropy distribution then is
\begin{align*}
  f(x)=\hat p_{q(x)}m
\end{align*}
where $q(x)$ is the interval that contains $x$.
We need to use this distribution to code $x^M$. Let $q^M$
be the quantization of $x^M$. The coding can then be
done as follows: first transmit $\hat p_j$ 
(the type of $q^M$), and then which specific 
sequence $q^M$ is in that type class; it is now known
in which interval $x_i$ is, and only the specific value
has to be transmitted with $-\log m$ bits. We can then
write the total codelength as
\begin{align}
  \hat L&=\hat L_q-M\log m
   \label{eq:hatL}
\end{align}
where $\hat L_q$ is the codelength to encode $q^M$ with a universal
coder (e.g., as outlined above). If $m\ll M$ this is \cite{Shamir06}
\begin{align*}
  \hat L_q&=MH(\hat p)+\frac{m-1}{2}\log M+O\left(\frac 1 M\right)
\end{align*}

The question is how to choose $m$. It is clear that it is necessary to
have unbounded $m$ to catch arbitrary deviations from the uniform
distribution, and the following theorem shows this is also sufficient to some extent
\begin{theorem}\label{thm:MDLuniversal}
	Suppose that the test data was generated by a continuous distribution. 
	Then the
	histogram detector  satisfies $P_{FA}\to 0$ and
	$P_D\to 1$ as $M\to\infty$ with suitable choice of
	$m\to\infty$.
\end{theorem}
\begin{proof}
    Let $x$ be a point where $P(x)\neq \widehat P(x)$, lets say $P(x)< \widehat P(x)$.
Since $P, \widehat P$ are continuous, there exists a neighborhood $N(x)$ so that
$\forall x'\in N(x):P(x')< \widehat P(x')$.
It then follows that $\int_{N(x)} P(x')dx'<\int_{N(x)} \widehat P(x')dx'$.
We now choose $K$ so large that some histogram bin, call it $i$, is completely within $N(x)$. Then $p_i \neq \tilde p_i$, where $\tilde p_i$ is the probability of the $i$-th bin according to $\widehat P$. 
It follows that $D(p\|\tilde p)>0$. Since $D(p\|\hat p)\stackrel{D}{\to} 0$ under
$H_0$ and $D(p\|\hat p)\stackrel{D}{\to} D(p\|\tilde p)>0$ under $H_1$, the theorem follows.
\end{proof}
Still, this leaves open how exactly to choose $m$. Without some type of
model selection (e.g., MDL) the only choice is to let $m$ be some
function of $M$. The
issue here is that there are contradictory requirements on $m$. First,
as we have to let $m\to\infty$, for large $M$ we get a large $m$,
which is not always advantageous according to Theorem \ref{thm:complexity}. 
To
detect smooth distributions one requires $m\ll M$, while to detect
concentrated distributions one would like $m>M$, as argued below.
 A traditional
MDL approach for selecting $m$ (e.g., \cite{HallHannon88,RissanenSpeedYu92})
would choose $m$ to get a good histogram approximation, e.g., it would
not give $m>M$, which is desired for OOD.
On the other hand, with (\ref{eq:MDLtest}) or (\ref{eq:WMDL}) there are
no restrictions on $m$ -- we will see that it is indeed possible to
have an infinite sum in (\ref{eq:WMDL}). One advantage of the coding
approach is that it allows "degenerate" models with $m>M$.

In order to put this in a theoretical framework, we consider the case of a very concentrated alternative distribution.
To detect this one does not need a large number of samples. If one
has a few samples close together, this is a strong indication that
the distribution is not uniform. This can be detected by a histogram
with small bin size, i.e., large $m$. For theoretical analysis we will consider the
extreme case of this, where the alternative distribution has a discontinuous
CDF (i.e., discrete or mixed). There is then a good chance that two
samples are identical, and that is an definite indication that the
distribution is not uniform. We will show that the coding approach
is able to detect this.

We will first argue that this is not possible without the coding combining
 in the sense that the false alarm probability
is one no matter how large $\tau$. Suppose that data is from the default
distribution (i.e., uniform) and $m$ is so large that all
samples are in different bins. Then the negative
log-likelihood according to (\ref{eq:hatL}) is
\begin{align*}
  \hat L &= M\log M-M\log m
\end{align*}
(when $m\gg M$ a more efficient coder is to transmit the
sequence $q^M$ uncoded)
which is unbounded (negative) as $m \to\infty$, i.e.,
for some $m$, $\hat L+\tau <0$ no matter how large $\tau$ or $M$. Thus, without coding one has to limit
$m<\infty$, and then one cannot detect identical samples for finite $M$.

On the other hand, with well designed coding we get the following
result
\begin{theorem}
	For a histogram detector with unbounded number of bins there is 
	a universal coder so that
\begin{enumerate}
	\item For the detector using (\ref{eq:MDLtest})
	\begin{itemize}
  \item For sufficiently large $\tau$ and/or $M$, the false alarm
  probability can be made arbitrarily small.
  \item If $x^M$ has at least \emph{three} non-unique samples (a sample repeated three times
  or two values repeated once), $x^M$ will be classified
  as OOD with probability one.
\end{itemize}
	\item For the  weighted detector (\ref{eq:WMDL})
	\begin{itemize}
  \item For sufficiently large $\tau$ and/or $M$, the false alarm
  probability can be made arbitrarily small.
  \item If $x^M$ has at least \emph{two} non-unique samples  $x^M$ will be classified
  as OOD with probability one.
\end{itemize}

\end{enumerate}

\end{theorem}

\begin{proof}
The paper \cite{Shamir06}
has a coder for $q^M$ for $m>M$, but for $m\gg M$ a more
efficient scheme is possible. The encoding is
done as follows
\begin{enumerate}
  \item The encoder encodes the number $b$ of
  bins that have at
  least one sample; this requires $\log M$ bits.
  \item The encoder encodes which bins have at least
  one sample. This requires $\log\binom{m}{b}$ bits.
  \item The encoder transmits which bin each sample is in;
  this is now transmission of $M$ samples from
  a $b$ alphabet source.
\end{enumerate}
The only difference from \cite{Shamir06} is that it makes it
\emph{explicit} that there might be less than $M$ non-zero bins. Since the decoder then knows
exactly the number of non-zero bins, if $b=M$ the encoder does not have to
encode the distribution over bins, which is $\frac 1 M$ for all
bins. But in the coding scheme of \cite{Shamir06} the encoder
always has to encode the distribution over bins
since the decoder does not know the number of
non-zero bins from the first step. The conclusion is
that the coding schemes has the same redundancy and regret.

The total codelength is
\begin{align*}
  \hat L  &= \log^*m+\log\binom{m}{b}-M\log m + \kappa(M) \nonumber\\
     &\geq \log^*m+mH\left(\frac b m\right)
     +\frac 1 2\log\frac{m}{8b(m-b)}-M\log m + \kappa(M) \nonumber\\
     &=\log^*m+b\log m-M\log m+ \kappa(M)
\end{align*}
Here $\kappa(M)$ denote terms that are bounded in $m$,
e.g. $O(1)$.
We see that if $b\leq M-2$ this can be negative
for sufficiently large $m$ no matter what is $\kappa(M)$,
but is always strictly positive for $b> M-2$. We next calculate
$P(b\leq M-2)$, as follows
\begin{itemize}
  \item We pick $M-2$ bins. This can be done $\binom{m}{M-2}$ ways.
  \item We distribute the $M$ samples in the $M-2$ bins,
  where some are allowed to be empty. This can be done
  $\binom{k+k-2-1}{k-3}$ ways.
\end{itemize}
Then
\begin{align*}
  P(b\leq M-2) &= \frac{m!}{(m-M+2)!(M-2)!}\cdot\frac{(2M-3)!}{M!(M-3)!}\cdot \frac{1}{m^M} \nonumber\\
  &= 1\cdot\left(1-\frac 1 m\right)\cdots \left(1-\frac {M-2} m\right)\frac{1}{m^2}\cdot\frac{(2M-3)!}{M!(M-3)!} \nonumber\\
  &\leq \frac 1{m^2}f(M)
\end{align*}
That means that the union bound over
 $m$ is finite, which again means
that the false alarm probability can be made arbitrary small
for sufficiently large $\tau$ or $M$.

Consider instead the MDL-weighted coder. The criterion is
\begin{align*}
  \sum_{m=m_0}^\infty 2^{-\hat L}
  &=\sum_{m=m_0}^\infty m^{M-b(m)}2^{-\log^* m}
  >C
\end{align*}
where $C$ is some constant the depends on $\tau$ and $M$. Under the
default uniform distribution $b(m)=M$ for sufficiently large $m\geq m_0'$ with
probability 1. But the sum $\sum_{m=m_0'}^\infty 2^{-\log^* m}$ is finite,
and for $C$ large enough the false alarm probability can be made
arbitrarily small.

On the other hand, if at least one value is repeated, $b(m)<M$ for
arbitrary large $m>m_0'$. But then the sum
\begin{align*}
  \sum_{m=m_0}^\infty m^{M-b(m)}2^{-\log^* m} \geq 
  \sum_{m=m_0'}^\infty m2^{-\log^* m}
  =\sum_{m=m_0'}^\infty 2^{-\log^* m+\log m}=\infty
\end{align*}
Therefore, no matter how large $C$ is, the data will be detected
as anomalous.

\end{proof}
The theorem also shows that the weighted detector (\ref{eq:WMDL}) 
can be strictly better than the "model selection" detector 
(\ref{eq:MDLtest}); in terms of codelength (\ref{eq:WMDL})
was already known to be better. At the same time, the proof of the theorem is
based on a carefully designed coder, one that is optimum in the
minimax coding sense. This indicates that using better coding in general -- better coding meaning
coders with shorter codelength -- results in better detection.

\section{Transformations}
One important detail in Martin-L\"{o}f-Kolmogorov randomness detection
is that the universal Turing machine implementing Kolmogorov complexity
has as input also the default distribution $P$. In many cases this
disappears asymptotically, but not in more complicated setups \cite{LiVitanyi}. Intuitively,
this also makes sense: If $P$ is very complicated, as universal coder
without any knowledge of $P$ would  have to estimate $P$ before
it can start detecting deviations from $P$ -- requiring a large number of samples.
On the other hand, it is difficult
to use $P$ in a universal source coder. Furthermore, in real-world implementation,
$P$ is not known, but found through machine learning. One option for using
knowledge of $P$ is 
to modify the ML distribution through some sort of online learning. However, this
is not very feasible, and does not fit into the statistics framework
we use. Therefore, the approach we take is to transform any distribution
into a standard distribution, and then apply the coding approach. For
known $P$ this is alway possible, as follows

\begin{lemma}[\cite{SabetiHost17}]
\label{Generate.thm}For any continuous random variable $\mathbf{X}$
there exists an $n$-dimensional uniform random variable $\mathbf{U}$,
so that $\mathbf{X}=\check{\mathbf{F}}^{-1}(\mathbf{U})$.
\end{lemma}

For unknown $P$ we use some type of bidirectional generative networks,
described in more detail below.

This approach has several advantages
\begin{itemize}
  \item Knowledge of $P$ is utilized in the universal coder.
  \item Since the default distribution is always the same, a standard set
    of statistics can be used.
  \item Often small deviations from the default models can be
    captured by simple models, which is advantageous according to
    Theorem \ref{thm:complexity}.
\end{itemize}

A simple example is given by the histogram approach in
Section \ref{sec:histogram}. If $P$ is very complex (e.g.,
with high frequency content) and the data is generated 
 by a distribution very close to $P$, one would need
very large $m$ to detect this. On the other hand, if
$P$ is known, data can be transformed with the CDF, and
even $m=2$ might be able to detect the deviation.

\section{Multivariate Gaussian Default Model, $P$}\label{sec:Gaussian}
We consider the case when the default model is Gaussian with zero mean
and (known) covariance matrix $\boldsymbol{\Sigma}$; mostly we consider
the case $\boldsymbol{\Sigma}=\mathbf{I}$. The aim is to find some statistics
that captures deviation from this model well. As mentioned in ??, the best
statistics have the default model as maximum entropy distribution for some
value of the statistic. The most obvious statistic is of course the
mean and covariance, ${\mathbf{T}}(\mathbf{x})=(\mathbf{x},\mathbf{x}\mathbf{x}^T)$;
the corresponding maximum entropy distribution is Gaussian $\mathcal{N}(\hat{\boldsymbol{\mu}},\hat{\boldsymbol{\Sigma}})$.
However, this is a high complexity model, which should not be used alone
according to Theorem \ref{thm:complexity}. We therefore consider lower
complexity models by specifying a sparse covariance matrix by requiring
${\boldsymbol{\Sigma}}^{-1}_{i,j}=0$ for some coordinates $(i,j)\in J$
and putting ${\boldsymbol{\Sigma}}_{i,j}=\hat{\boldsymbol{\Sigma}}_{i,j}$
for $(i,j)\notin J$. There exists a unique positive definite matrix $\mathbf{S}$ 
satysfying these constraints, and the maximum entropy distribution is
the Gaussian distribution with covariance matrix $\mathbf{S}$ \cite{Dempster72}.
The method is called covariance selection \cite{Dempster72}.

The minimum size of the covariance statistic that gives a valid maximum
entropy distribution is the dimension of $\mathbf{x}$ (by only estimating
the diagonal elements), which can still be
high. We can also consider simpler statistics. One can start with 
$r^2=\mathbf{x}^T\mathbf{x}$, and then consider statistics ${\mathbf{T}}(r^2)$.
It is clear that the maximum entropy distribution corresponding to ${\mathbf{T}}(r^2)$
is uniformly distributed over an $n$-ball. Thus the maximum entropy
distribution is
\begin{align*}
  f(\mathbf{x}) &= \frac {\Gamma(n/2)}{\pi^{n/2}r^{n-2}}f_r(r^2)=\frac {\Gamma(n/2)}{\pi^{n/2}(r^2)^{n/2-1}}\exp\left(\lambda_0-1+\sum_{i=1}^m \lambda_i t_i(r^2)\right)
\end{align*}
For the default distribution, $r^2$ is $\chi^2$ distributed. As discussed
previously, it is advantageous to have the default distribution to be
a special case of the maximum entropy distribution. We therefore use
${\mathbf{T}}(r^2)=(r^2,\ln(r^2))$ giving a Gamma maximum entropy distribution,
\begin{align}
  f(\mathbf{x}) &= \frac {\Gamma(n/2)}{\pi^{n/2}(r^2)^{n/2-1}}
  \frac{\beta^\alpha}{\Gamma(\alpha)}(r^2)^{\alpha-1}
  \exp(-\beta r^2) \nonumber\\
  &= \frac {\Gamma(n/2)}{\pi^{n/2}}
  \frac{\beta^\alpha}{\Gamma(\alpha)}(r^2)^{\alpha-n/2}
  \exp(-\beta r^2)
\end{align}


The $r^2$ statistic detects deviations in the radial distribution of the
test data. This can be complemented by detecting deviations in directional
distribution. A natural statistic is ${\mathbf{T}}(\mathbf{x})=\frac{\mathbf{x}}{\|\mathbf{x}\|}$.
The corresponding maximum entropy distribution is the von Mises-Fisher
distribution (for unit vectors $\mathbf{u}$)
\begin{align*}
  f_d(\mathbf{u})=C(\kappa)\exp(\kappa\boldsymbol{\mu}^T\mathbf{u})
\end{align*}
where $C(\kappa)$ is a normalization constant, and
\begin{align*}
  \boldsymbol{\mu} & = E[\|\mathbf{u}\|] \nonumber\\
  \frac{I_{m/2}(\kappa)}{I_{m/2-1}(\kappa)} &= \|E[\|\mathbf{u}\|]\|
\end{align*}
The Gaussian distribution with $\boldsymbol{\Sigma}=\mathbf{I}$ gives
a von Mises-Fisher
distribution with $\kappa=0$.
The total distribution for $\mathbf{x}$ is then
\begin{align*}
  f(\mathbf{x}) &= f_r(\|\mathbf{x}\|^2)f_d\left(\frac{\mathbf{x}}{\|\mathbf{x}\|}\right)
  \frac 1{\|\mathbf{x}\|^{n-2}}
\end{align*}

One can think of it as follows. The (sparse) covariance statistic can
detect deviations in covariance, but not a deviation that has  a
covariance matrix $\mathbf{I}$, but a non-Gaussian distribution. The
$r^2$ statistic and directional statistic can detect deviations from a
Gaussian distribution. For example, if the components of $\mathbf{x}$
are IID, but not Gaussian, the covariance matrix is $\mathbf{I}$, but
$r^2$ is not $\chi^2$.

%

\subsection{OOD under multivariate Gaussian default model}
As outlined, the default model $P$ is assumed to be a known multivariate Gaussian distribution. The model for out-of-distribution data is also assumed to be Gaussian but with unknown covariance matrix $\mathbf{\Sigma}$. Without loss of generality, we assume that everything is zero mean.
We calculate CTW-based weighting criterion introduced in \eqref{eq:WMDL} to find if a batch of data $\mathbf{x}^M$ belongs to $P$ or is OOD. In order to compute \eqref{eq:WMDL}, we follow these steps:

\begin{enumerate}
\item Encode the data $\mathbf{x}^M$ with the known default model $P$. Therefore $L = -\log P(\mathbf{x}^M)$.

\item Encode the data $\mathbf{x}^M$ with universal multivariate Gaussian coder for all unique sparsity patterns obtained from covariance matrix estimation. This is described in Section \ref{sec:coding}.

\item Encode the data $\mathbf{x}^M$ with universal Gamma distribution to account for $r^2=\mathbf{x}^T\mathbf{x}$ statistics. This is described in Section \ref{sec:GammaCoding}.

\item Combine codelengths from step--2 and step--3 using CTW principle in \eqref{eq:WMDL} to get $\widehat{L}$. 

\item Given a threshold $\tau$, the data $\mathbf{x}^M$ is OOD if $\widehat{L}+\tau < L$.
\end{enumerate}



\subsection{Universal Multivariate Gaussian Coder}
\label{sec:coding}

A universal multivariate Gaussian coder was
proposed in \cite{abolfazli2021graph}; it is universal in the sense that it can be used to encode any multivariate Gaussian data. Our approach to finding the description length of a multivariate Gaussian model is based on characterizing the distribution by the sparsity pattern of the inverse covariance matrix, $\boldsymbol{\Sigma}^{-1}$. This sparsity pattern is known as the conditional independence graph, $G$, of the Gaussian. It can be found by using a number of structure learning methods such as graphical lasso (GLasso) \cite{friedman2008sparse}; these methods often use a regularization parameter, $\lambda$, to control for the sparsity of the solution. Each value of $\lambda$ is associated with a conditional independence graph $G$. Here, we want to combine codelength of unique models and as a consequence, we consider unique conditional independence graphs.




The codelength of the universal multivariate Gaussian coder is the sum of two components: 
 1) $L(G)$, the number of bits needed to describe a conditional independence graph $G$ which can be found using any graph coder described in~\cite{Abolfazli21ISIT}. 
     2) $L(\mathbf{x}^M|G)$, the number of bits to describe $\mathbf{x}^M$ using a multivariate Gaussian with conditional independence graph $G$. Any coding scheme that achieves the universal coding lower bound can be used. We chose predictive MDL \cite{Rissanen86} because it gives the actual codelength of data:

\begin{equation}
   L(\mathbf{x}^M|G) = -\sum_{i=1}^{M-1} \log_2 P\left(\mathbf{x}_{i+1} ; \hat{\theta}(\mathbf{x}_1,\ldots, \mathbf{x}_i)\right) \label{eq:predMDL}
\end{equation}
where $\hat{\theta}(\mathbf{x}_1,\ldots, \mathbf{x}_i)$ is the maximum likelihood estimate of the covariance matrix computed using the first $i$ samples, $\{\mathbf{x}_1, \ldots, \mathbf{x}_i \}$. The solution of the estimate is constrained such that the corresponding conditional independence graph is $G$. In \cite[Section 17.3]{hastie2017elements}, an iterative method was proposed to solve this constrained problem.

Algorithm~\ref{alg:mulgau} describes the implementation of the universal multivariate Gaussian coder.

\begin{algorithm}[h!]\caption{\textsc{Universal Multivariate Gaussian Coder}}
   \label{alg:mulgau}
\begin{algorithmic}[1]
\State {{\bfseries Input:} IID zero-mean multivariate Gaussian data $\mathbf{x}^M$, GLasso regularization parameters $\{\lambda_1,\ldots, \lambda_K\}$}
\State $\mathcal{G} = \{ \} $ \Comment{set of unique models}
\State $\mathcal{L} = \{ \} $ \Comment{set of computed codelengths}
\For {Each regularization parameter $\lambda \in \{\lambda_1,\ldots, \lambda_K\}$.}
	\State {Use GLasso with $\lambda$ to find $G_\lambda$}
    \If {$G_\lambda \notin \mathcal{G}$}
        \State $\mathcal{G} \gets \mathcal{G} \cup G_\lambda$
	    \State {Compute $L(G_\lambda)$ using any of graph coders described in \cite{abolfazli2021graph}}
	    \State {Compute $L(\mathbf{x}^M|G_\lambda)$ using predictive MDL in \eqref{eq:predMDL}}
        \State $L = L(G_{\lambda}) + L(\mathbf{x}^M|G_\lambda)$
        \State $\mathcal{L} \gets \mathcal{L} \cup L$
    \EndIf
\EndFor
\State {\bfseries Output:} $\mathcal{L}$
\end{algorithmic}
\end{algorithm}

\subsection{Universal Gamma Coder}
\label{sec:GammaCoding}
For ${\mathbf{T}}(r^2)$ statistics, we consider Gamma distribution as the maximum entropy distribution where $\chi^2$ is a special case of it. In order to encode $\mathbf{x}^M$ using Gamma distribution, we use predictive MDL where we use maximum likelihood to estimate the parameters of Gamma distribution (shape and scale) sequentially.

\subsection{Experiments on Synthetic Data}\label{sec:synthetic}
Although our approach, maximum entropy coding (MEC), is designed when data is normally distributed, the experiments show that it also performs well when the data comes from distributions that are marginally near-Gaussian.
We compared our approach (MEC) to $d-$dimensional KS test (\textup{ddKS}) method \cite{hagen2021accelerated}, a multi-dimensional two-sample KS test, for 6 different test scenarios of synthetically generated multivariate Gaussian and nearly-Gaussian data as described in Table~\ref{tab:synth_scenarios}.
For \textsc{Case-1\&2}, the distribution of both the default model, $P$ and the alternative model, $\widehat{P}$ are multivariate Gaussian. For \textsc{Case-3} to 6, the test data is generated from linear transformation of nearly-Gaussian distributions, $\mathbf{Ax}$, where $[\mathbf{A}_{ij}]$ is the transformation matrix.


For each test case, we performed OOD detection on a synthetically generated test dataset of size $M$. The test set, $\mathbf{x}^M$, is either generated from the default model, $P$ or from the alternative model, $\widehat{P}$. We repeated the experiment $1000$ times; in 500 experiments, data is drawn from default distribution and in 500 experiments, data is generated by alternative distribution. The AUROC for the six scenarios is shown in Table~\ref{tab:synth_results}. 
As it can be seen, our approach outperforms \textup{ddKS} method in all cases, even for the cases where the test data do not exactly come from Gaussian distributions.


\begin{table*}[h]
\centering
  \begin{threeparttable}[b]
  \caption{Scenarios for generating synthetic test data.}
   \begin{tabular}{lccc}
        \toprule
        & Parameters of Default model, $P$ & Parameters of Anomalous model, $\widehat{P}$    & Data generation\\
        \midrule
        \vspace{0.1cm}
        \textsc{Case--1} &   \begin{tabular}{@{}l@{}}$\mathbf{\Omega}_{ii} = 1, \mathbf{\Omega}_{i,i-1} = \mathbf{\Omega}_{i-1,i} = 0.45$ \\ $\mathbf{\Omega}_{16} = \mathbf{\Omega}_{61} = 0.45$\end{tabular}  &   $\mathbf{\Omega}_{ii} = 1, \mathbf{\Omega}_{i,i-1} = \mathbf{\Omega}_{i-1,i} = 0.45$   & $\mathbf{X} \sim \mathcal{N}(\mathbf{0}, \mathbf{\Omega}^{-1})$ \\
        
        \vspace{0.1cm}
        \textsc{Case--2} &    \begin{tabular}{@{}l@{}}$\mathbf{\Omega}_{ii} = 1, \mathbf{\Omega}_{i,i-1} = \mathbf{\Omega}_{i-1,i} = 0.45$ \\ $\mathbf{\Omega}_{16} = \mathbf{\Omega}_{61} = 0.45$\end{tabular}     &  \begin{tabular}{@{}l@{}}$\mathbf{\Omega}_{ii} = 1, \mathbf{\Omega}_{i,i-1} = \mathbf{\Omega}_{i-1,i} = 0.5$ \\ $\mathbf{\Omega}_{i,i-2} = \mathbf{\Omega}_{i-2,i} = 0.25$\end{tabular}    &  $\mathbf{X} \sim \mathcal{N}(\mathbf{0}, \mathbf{\Omega}^{-1})$     \\

        \vspace{0.1cm}
        \textsc{Case--3}    &  \begin{tabular}{@{}l@{}}$\mathbf{A}_{ii} = 1, \mathbf{A}_{i,i-1} = \mathbf{A}_{i-1,i} = 0.5$ \\ $\mathbf{A}_{i,i-2} = \mathbf{A}_{i-2,i} = 0.25$\end{tabular}     &   \begin{tabular}{@{}l@{}l@{}}$\mathbf{A}_{ii} = 1, \mathbf{A}_{i,i-1} = \mathbf{A}_{i-1,i} = 0.4$ \\ $\mathbf{A}_{i,i-2} = \mathbf{A}_{i-2,i} = 0.2$ \\ $\mathbf{A}_{i,i-3} = \mathbf{A}_{i-3,i} = 0.2$ \end{tabular}     &  $x_i \sim \textrm{Laplace}(0, i)$ \\

        \vspace{0.1cm}
        \textsc{Case--4}    &  same as \textsc{Case--3}   &   same as \textsc{Case--3}     &   $x_i \sim \textrm{Logistic}(0, i)$    \\

        \vspace{0.1cm}
        \textsc{Case--5}    &  same as \textsc{Case--3}   &   same as \textsc{Case--3}     &    $x_i \sim\ \chi^2_{i+4}$    \\

        \textsc{Case--6}    &  same as \textsc{Case--3}   &   same as \textsc{Case--3}     &    $x_i \sim \textrm{StudentT}(i+4)$    \\
        
        \bottomrule
        \label{tab:synth_scenarios}
    \end{tabular}
    \end{threeparttable}
    \vspace{-0.3in}
\end{table*}    

\begin{table}[htbp]
  \centering
  \caption{AUROC comparing our approach (MEC) to the multi-dimensional KS test in \cite{hagen2021accelerated} (ddKS). 
  }
   \begin{tabular}{l cc| cc}
        \toprule
        & \multicolumn{2}{c|}{$\mathbf{M=25}$} & \multicolumn{2}{c}{$\mathbf{M=50}$} \\
        
        & MEC & \textup{ddKS} &  MEC & \textup{ddKS}  \\
        \midrule
        \textsc{Case--1} & $\mathbf{0.957}$ &  $0.824$ &  $\mathbf{0.985}$   &   $0.939$    \\
        \textsc{Case--2} & $\mathbf{0.999}$ &  $0.987$ &  $\mathbf{1.0}$   &   $\mathbf{1.0}$    \\
        \textsc{Case--3} & $\mathbf{0.980}$ &  $0.553$ &  $\mathbf{0.994}$   &  $0.582$     \\
        \textsc{Case--4} & $\mathbf{0.984}$ &  $0.564$ &  $\mathbf{0.994}$  &   $0.594$    \\
        \textsc{Case--5} & $\mathbf{0.920}$ &  $0.563$ &  $\mathbf{0.944}$   &   $0.591$    \\
        \textsc{Case--6} & $\mathbf{0.983}$ &  $0.707$ &  $\mathbf{0.991}$   &   $0.888$    \\
        
        \bottomrule
        \label{tab:synth_results}
    \end{tabular}
     \vspace{-0.3in}
\end{table}

\section{Unknown Default Model, $P$}\label{sec:simple}


In the previous section, we have shown that our coding-based OOD detection approach works well for multivariate Gaussian and near-Gaussian distributions. However, most real-world data are far from Gaussian. In fact, the default model is not known for real-world data.
We overcome this by using a (non-linear) continuous transform so that the data is Gaussian in the transformed or the latent space. In this paper, we used generative neural networks to transform arbitrary data to multivariate Gaussian. Our requirements are that 1) the transformation, $\mathbf{z} = f(\mathbf{x})$, from data the space, $\mathbf{x} \in \mathbb{R}^n$, to latent space, $\mathbf{z} \in \mathbb{R}^m$, is invertible. This means that there is a function $g$ such that $\mathbf{x} = g(\mathbf{z}) = f^{-1}(\mathbf{z})$; 2) the distribution in the latent space can be specified (usually as a multivariate Gaussian).


For the transformation, we considered Glow \cite{kingma2018glow}, a flow-based generative network.
Glow is exactly invertible. Our method to solve OOD detection problem is described in Algorithm \ref{alg:unknownglow}. Since the default distribution is unknown, training data is used to 1) train the Glow network (i.e., learn $g(\cdot)$ and $f(\cdot)$), 2) learn the multivariate distribution of the latent representation. Theoretically, this distribution can be specified a priori. However, in practice, we found that it is safer to estimate the distribution from data.

The disadvantage of Glow is that the latent space has to be the same dimension as the data space ($m = n$). This poses a limitation on our approach as it requires a number of matrix inversions of not necessarily well-conditioned matrices. We suggested downsampled the data before training Glow.

\begin{algorithm}[h!]\caption{\textsc{Glow algorithm}}
   \label{alg:unknownglow}
\begin{algorithmic}[1]
\State {\textit{Learn Default Model $P$}}
\State {\bfseries Input:} IID training data $\mathbf{x}^N$
\State Learn functions $f(\cdot)$ and $g(\cdot$) by training the neural network on $\mathbf{x}^N$. We used Glow; other invertible generative models can also be used
\State Estimate $\mathbf{\Sigma}_{train}$, the covariance matrix of $\mathbf{x}^N$ from the model giving the shortest codelength in Algorithm~\ref{alg:mulgau}
\State {\bfseries Output:} $f(\cdot)$ and $g(\cdot$) of the neural network, learned default model $P = \mathcal{N}(\mathbf{0}, \mathbf{\Sigma}_{train})$
\newline
\State {\textit{OOD Detection}}
\State {\bfseries Input:} Outputs from Training, IID test data $\mathbf{x}^M, \tau$
\State Find the latent representation $\{\mathbf{z}_i = f(\mathbf{x}_i)\}$

\State Compute  $L = -\sum_{i = 1}^M \log(P(\mathbf{z}_i))$, where $P$ is the learned default model

\State Compute $\widehat{L}$ by coding $\mathbf{z}^M$
with universal multivariate Gaussian coder and Gamma coder described in Section~\ref{sec:coding} and \ref{sec:GammaCoding}

\State {\bfseries Output:} Data $\mathbf{x}^M$ is OOD if $\widehat{L} + \tau < L$
\end{algorithmic}
\end{algorithm}

\section{Experiments on Real-World Data}
\label{sec:experiment}

%
We considered the digital image dataset MNIST \cite{deng2012mnist} where we do not know the default model distribution $P$ for the data. Instead, we have a set of training data $\mathbf{x}^N$.
We took the training data from the MNIST dataset and considered three sets of experiments for OOD detection:
\begin{itemize}
\item \textbf{Experiment 1:} Detect if a test set is from MNIST or fashion MNIST \cite{xiao2017}.
\item \textbf{Experiment 2:} Detect if a test set is from MNIST or non-MNIST \cite{bulatov2011notmnist}.
\item \textbf{Experiment 3:} Detect if a test set is from MNIST or synthetically-perturbed MNIST (see Table~\ref{tab:scenarios} and Figure~\ref{fig:Cases} for the description and visualization of the different dataset).
\end{itemize}

\begin{table}[htbp]
  \centering
  \caption{Scenarios for synthetically perturbing MNIST images. Rotation and shearing values are in degree, width and height shift in fraction, zoom and brightness in range.}
   \begin{tabular}{lc}
        \toprule
        & Perturbation type, Value \\
        \hline
        \textsc{Case--1} & Rotation, $5$ \\
        \textsc{Case--2}  & Shearing, $20$\\
        \textsc{Case--3} & $[$ Width shift , Height shift$]$, $[0.02,0.02]$\\
        \textsc{Case--4} &Zooming, $[0.8,1.2]$ \\
        \textsc{Case--5} &Zooming, $[1,1.1]$ \\
        \textsc{Case--6} &Zooming, $[0.9,1]$ \\
        \textsc{Case--7} &Brightness, $[0.2,2]$ \\
        \textsc{Case--8} &Brightness, $[0.2,1]$ \\
        \textsc{Case--9} &Gaussian noise, $\mu = 0, \sigma = 0.05$ \\
        \bottomrule
        \label{tab:scenarios}
    \end{tabular}
\end{table}


The training data consists of $60,000$ black and white images: $\{\mathbf{x} \in \mathbb{R}^{28 \times 28}\}$. In order to use Algorithm~\ref{alg:unknownglow} in the
current implementation, we had to downsample the image from $28\times28$ to $8\times 8$ pixels. This is because Glow can not do dimensionality reduction, and our approach needs to compute $\boldsymbol{\Sigma}^{-1}$. Inversion of high-dimensional matrices entails precision issues that can create invalid results; we are working on approximate matrix inversion to address this issue.

We solve OOD detection problem on test datasets of size $M$. We repeated the experiment $1000$ times; in 500 experiments, test data is from the same dataset as the training data), and in 500 experiments, test data is from a different dataset than the training data.
It should be mentioned that we can not compare to the ddKS method used in Section~\ref{sec:synthetic} because the data is too high-dimensional for the KS test.





\begin{figure}[h]
\centering
\hspace{-0.3in}
{\begin{tabular}{ccc}
 \subfloat[MNIST]{\includegraphics[width=0.9in]{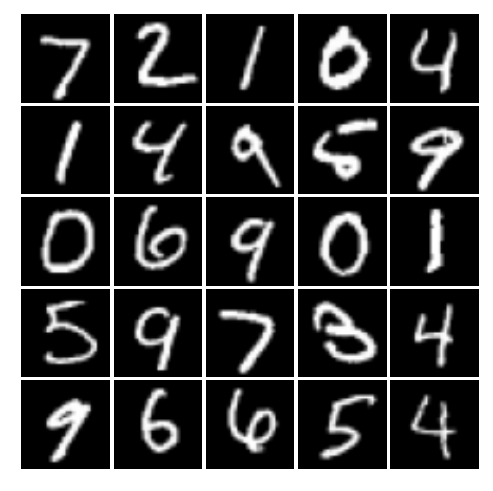}} &
 \subfloat[Fashion]{\includegraphics[width=1in]{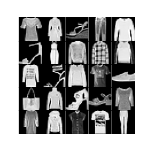}}
 &
 \subfloat[not-MNIST]{\includegraphics[width=1in]{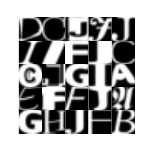}}\\
\subfloat[\textsc{Case--1}]{\includegraphics[width=0.9in]{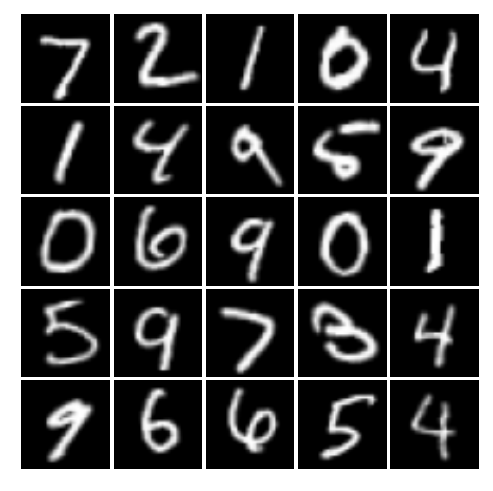}} & 
\subfloat[\textsc{ Case--2}]{\includegraphics[width=0.9in]{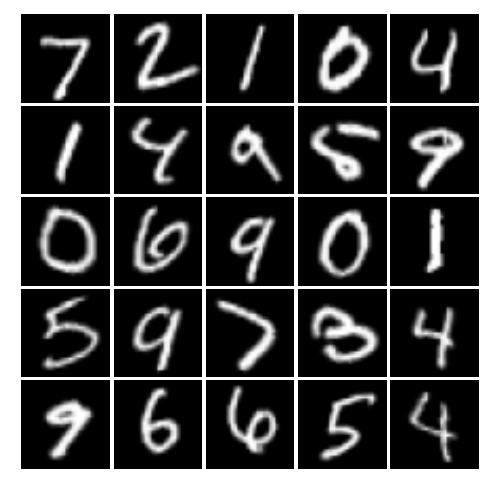}} & 
\subfloat[\textsc{Case--3}]{\includegraphics[width=0.9in]{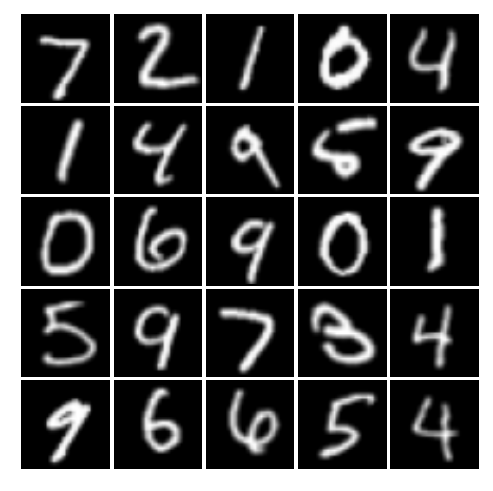}} \\
\subfloat[\textsc{Case--4}]{\includegraphics[width=0.9in]{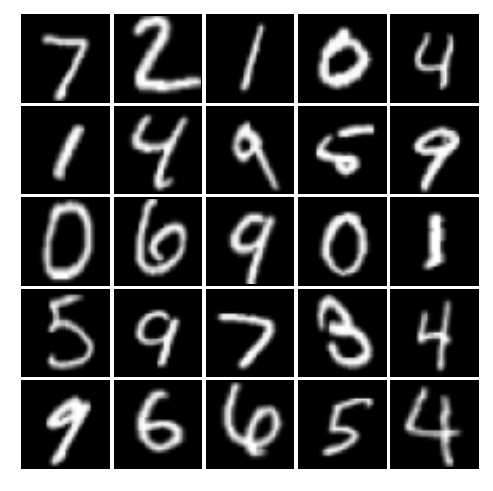}} & 
\subfloat[\textsc{Case--5}]{\includegraphics[width=0.9in]{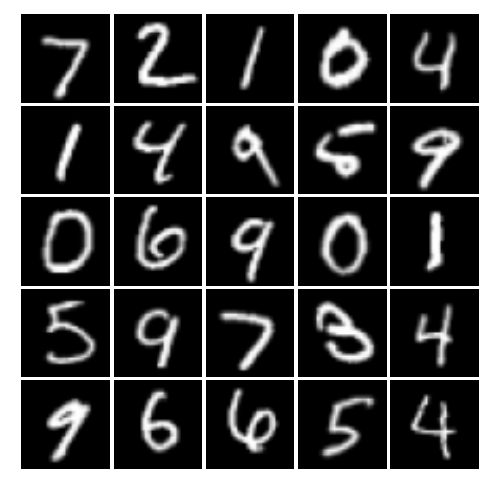}} &
\subfloat[\textsc{Case--6}]{\includegraphics[width=0.9in]{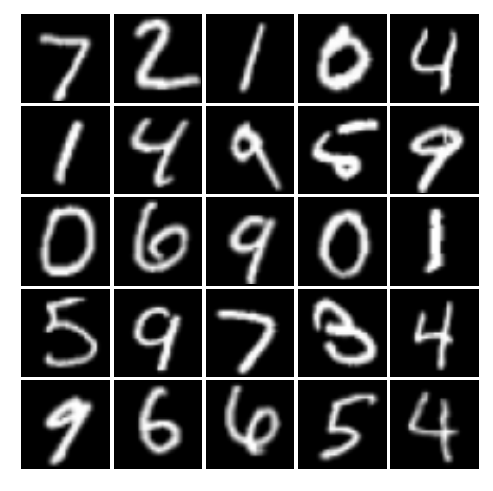}} \\
\subfloat[\textsc{Case--7}]{\includegraphics[width=0.9in]{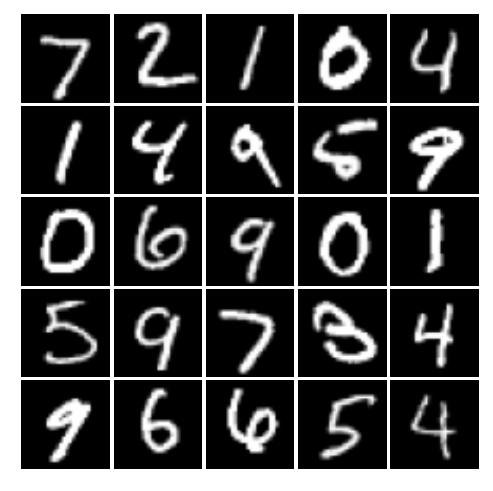}}&
\subfloat[\textsc{Case--8}]{\includegraphics[width=0.9in]{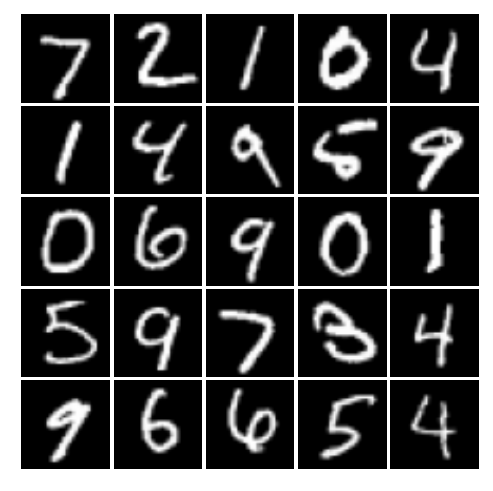}} &
\subfloat[\textsc{Case--9}]{\includegraphics[width=0.9in]{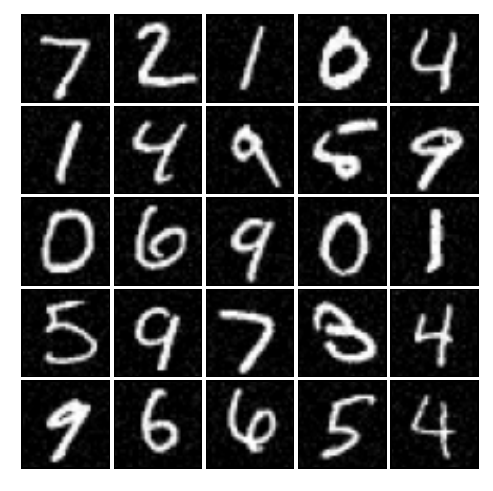}} \\
\end{tabular}}
\caption{Different datasets used in MNIST experiments. Note that the synthetically-perturbed images look very similar to the original.}
\label{fig:Cases}
\end{figure}

We compared our approach to another Glow-based method called Typicality \cite{NalisnickAl19}. We trained our model with the same hyperparameters and settings as \cite{NalisnickAl19}. For a fair comparison, we trained and tested Typicality with both downsampled $8 \times 8$ images.
Table~\ref{tab:results2} shows the AUROC for the experiments using different test set sizes $M$. Our method has higher performance than Typicality method in all cases except $\textsc{Case--5}$. In particular, in \textsc{Case--8}, Typicality gives an AUROC less than $0.5$ (i.e., random guessing). This is because (as noted in the introduction) a likelihood-based approach, which does not learn alternative distribution, may encounter situations where the OOD test data has a high likelihood under the learned default model.

\begin{table*}[hbpt!]
\centering
\caption{AUROC for MNIST  experiments comparing our MEC to Typicality \cite{NalisnickAl19} trained and tested on downsampled images. The best value for each case is boldfaced.}
\begin{tabular}{l cc| cc | cc }
\toprule
 & \multicolumn{2}{c}{$\mathbf{M=50}$} & \multicolumn{2}{c}{$\mathbf{M=100}$} & \multicolumn{2}{c}{$\mathbf{M=200}$}  \\
 & MEC  & Typicality  & MEC & Typicality  &  MEC & Typicality\\ 
\midrule
fashion MNIST  &  $\mathbf{1.000}$   & $\mathbf{1.000}$  &   $\mathbf{1.000}$ & $\mathbf{1.000}$ &   $\mathbf{1.000}$ &  $\mathbf{1.000}$  \\
not-MNIST  &  $\mathbf{1.000}$   & $\mathbf{1.000}$  &  $\mathbf{1.000}$ & $\mathbf{1.000}$  &    $\mathbf{1.000}$  &$\mathbf{1.000}$  \\
\textsc{Case--1}  &  $\mathbf{1.000}$   & $0.995$  &   $\mathbf{1.000}$  & $\mathbf{1.000}$ &  $\mathbf{1.000}$ &  $\mathbf{1.000}$ \\
\textsc{Case--2}  &  $\mathbf{1.000}$   &  $0.998$  &   $\mathbf{1.000}$  & $\mathbf{1.000}$ &  $\mathbf{1.000}$  & $\mathbf{1.000}$ \\
\textsc{Case--3}  &  $\mathbf{1.000}$   &  $\mathbf{1.000}$  &   $\mathbf{1.000}$  & $\mathbf{1.000}$ & $\mathbf{1.000}$ &  $\mathbf{1.000}$ \\
\textsc{Case--4}  &  $\mathbf{1.000}$   &  $0.502$  &   $\mathbf{1.000}$  & $0.505$ & $\mathbf{1.000}$  & $0.510$ \\
\textsc{Case--5}  &  $0.788$   &  $\mathbf{0.944}$   &   $0.855$  & $\mathbf{0.974}$ &  $0.911$  & $\mathbf{0.980}$ \\
\textsc{Case--6}  &  $\mathbf{1.000}$   &  $\mathbf{1.000}$   &   $\mathbf{1.000}$  & $\mathbf{1.000}$  &  $\mathbf{1.000}$ & $\mathbf{1.000}$ \\
\textsc{Case--7}  &  $\mathbf{0.978}$   &  $0.788$  &   $\mathbf{0.985}$  & $0.849$ &  $\mathbf{0.980}$ &  $0.911$ \\
\textsc{Case--8}  &  $\mathbf{0.883}$   &  $0.430$  &    $\mathbf{0.928}$  & $0.380$  &   $\mathbf{0.949}$ &   $0.315$ \\
\textsc{Case--9}  &  $\mathbf{1.000}$   &   $\mathbf{1.000}$ &   $\mathbf{1.000}$  & $\mathbf{1.000}$ &  $\mathbf{1.000}$ & $\mathbf{1.000}$ \\

\bottomrule
\end{tabular}%
\label{tab:results2}
\end{table*}

\section{Conclusion, Limitations, and
Future work}\label{sec:Conclusion}
We developed a method for OOD detection based on MDL for the case when the default distribution $P$ is known and when $P$ is unknown. In terms of application, the latter case is more likely to occur. We showed with experiments that our approach outperforms KS test for multivariate Gaussian and near-Gaussian OOD detection problems.  We tested our approach on MNIST dataset and compared with another OOD detection method called Typicality. Our approach has higher performance in most of cases. It also does not give AUROC less than 0.5.

In this paper, we restricted the class of potential data models to only multivariate Gaussian distributions. Since  the neural network transforms the default model into Gaussian, this is still powerful enough to detect subtle changes. However, if the test data is from distributions very different from the default distribution, it is unlikely that the transformed distribution will be close to Gaussian. The advantage of using MDL is that we can add any other model/coder into the framework as long as we account for complexity through MDL.
In the future, we plan to expand our method with mixture of Gaussians and an extension of the histogram approach in one dimension to higher dimensions.


We also have numerical challenges with our current implementation because our method needs to compute covariance matrix $\mathbf{\Sigma}$ and precision matrix $\mathbf{\Sigma}^{-1}$. High-dimensional data means that these matrices are large and precision issue means that the inversion results may not be valid covariance/precision matrices. 
We will investigate fast, approximate matrix inversion algorithms (e.g., \cite{benzi2000orderings}) to solve this problem.

\bibliographystyle{plainnat}
\bibliography{Coop06,ahmref2,Coop03,BigData,reference,refs}

\end{document}